\documentclass[10pt]{article}

\usepackage[T1]{fontenc}
\usepackage[utf8]{inputenc}
\usepackage{amsmath,amssymb,amsthm}
\usepackage{xcolor}
\usepackage{booktabs}
\usepackage{graphicx}
\usepackage{array} 
\usepackage{bm}
\usepackage[a4paper,left=0.85in,right=0.85in,top=0.95in,bottom=0.95in]{geometry}
\usepackage[expansion=false]{microtype}
\usepackage{titlesec}
\titlespacing*{\section}{0pt}{11pt plus 3pt minus 3pt}{5pt plus 1pt}
\titlespacing*{\subsection}{0pt}{7pt plus 2pt minus 2pt}{3pt plus 1pt}
\usepackage[hidelinks]{hyperref}
\hypersetup{pdftitle={WSVI: A Dimensionless Shape Family for Implied Volatility and its Static No-Arbitrage Structure},pdfauthor={Charles Clevenger, Xiang Wan}}

\usepackage{xcolor}
\definecolor{labelkey}{rgb}{0,0,1}
\definecolor{Red}{rgb}{0.7,0,0.1}
\definecolor{Green}{rgb}{0,0.7,0}

\usepackage{tocloft}
\newcommand{\R}{\mathbb{R}}
\newcommand{\sca}{\varsigma}
\newcommand{\dens}{\varrho}
\newcommand{\Ndist}{\mathcal{N}}
\newcommand{\ndens}{\mathfrak{n}}
\newcommand{\Cdf}{\mathcal{C}}
\newcommand{\Ray}{\mathcal{R}} \newcommand{\bpi}{\bm{\pi}}

\theoremstyle{plain}
\newtheorem{theorem}{Theorem}[section]
\newtheorem{proposition}[theorem]{Proposition}

\newtheorem{corollary}[theorem]{Corollary}

\theoremstyle{definition}
\newtheorem{definition}[theorem]{Definition}

\theoremstyle{remark}
\newtheorem{remark}[theorem]{Remark}

\AtBeginDocument{%
  \setlength{\abovedisplayskip}{9pt plus 2pt minus 3pt}%
  \setlength{\belowdisplayskip}{9pt plus 2pt minus 3pt}%
  \setlength{\abovedisplayshortskip}{3pt plus 2pt minus 2pt}%
  \setlength{\belowdisplayshortskip}{3pt plus 2pt minus 2pt}%
  \setlength{\topsep}{4pt plus 1pt minus 2pt}%
}

\allowdisplaybreaks
\numberwithin{equation}{section}

\title{\Large
    WSVI: A Dimensionless Shape Family for Implied\\[4pt]
    Volatility and Its Static No-Arbitrage Structure
}

\author{Charles Clevenger
             \thanks{Department of Mathematics and Statistics, Loyola University Chicago, Chicago, IL \hfill \texttt{cclevenger1@luc.edu}}
        \and
        Xiang Wan
             \thanks{Department of Mathematics and Statistics, Loyola University Chicago, Chicago, IL \hfill \texttt{xwan1@luc.edu}}
}

\date{\today}

\begin{document}

\maketitle

\begin{abstract}
W-shaped smiles appear in near-expiry options around binary events such as earnings, and have been associated with bimodal risk-neutral densities. The three-parameter eSSVI slice cannot produce them. This paper defines WSVI, a parametric family for implied volatility that admits negative at-the-forward curvature and bimodal implied densities, and develops its static no-arbitrage structure. The construction factorizes total variance into a level and a dimensionless shape of normalized log-moneyness. The shape extends the per-slice eSSVI form with bounded one-sided basis terms, which add flexibility in the interior while leaving the leading-order wing behavior controlled by the affine and quadratic components. We characterize the family's exact domain and write the butterfly, vertical spread, and calendar conditions directly in shape coordinates.
\end{abstract}

\vspace{0.8em}

{\scriptsize
  \tableofcontents
}

\section{Introduction}\label{sec:intro}
At each expiry, an options market quotes prices across a range of strikes. Converted to Black implied volatilities, these quotes form a smile whose shape reflects the market-implied distribution of the underlying at that expiry \cite{BreedenLitzenberger1978}. Models in the SVI family provide low-dimensional parametric representations of the implied-volatility smile at a single expiry. Their closed forms permit analytic differentiation and controlled extrapolation beyond the listed strikes. Their parameterizations also allow static no-arbitrage conditions to be studied analytically \cite{GatheralJacquier2014}.

Several extensions and reformulations of SVI have been developed, including SVI5, SVI-JW, and SSVI \cite{GatheralJacquier2014}, as well as extended-SSVI (eSSVI) \cite{HendriksMartini2019}. Here, the relevant baseline is the per-slice eSSVI form, which WSVI contains as its three-parameter base member. The limitation relevant here is structural. In the three-parameter eSSVI slice, the same parameters governing the core shape also determine the asymptotic behavior of the wings.

We seek a family with three properties: (i) enough flexibility to represent shapes unavailable to the three-parameter eSSVI slice, (ii) a parameterization in which static no-arbitrage conditions remain analyzable, and (iii) controlled behavior beyond the liquid strike region. These goals compete with one another since additional parameters increase flexibility while introducing new directions in which arbitrage and wing behavior must be controlled.

The SVI family and its surface form SSVI \cite{GatheralJacquier2014} are standard constructions, with eSSVI \cite{HendriksMartini2019} providing a further extension. Their no-arbitrage domains have been studied in detail \cite{Klassen2016,MartiniMingone2022,Mingone2022}. The per-slice eSSVI form carries total variance $\vartheta$, correlation $\rho_*\in(-1,1)$, and curvature $\phi>0$. It satisfies the second and third requirements well, but it fails the first in one structural respect. As Corollary~\ref{cor:convexity} shows, its shape is strictly convex at every strike for every admissible parameter choice. This prevents it from representing W-shaped smiles of the kind observed around binary events such as earnings releases, as studied by Glasserman and Pirjol \cite{GlassermanPirjol2023}. At zero skew, a local maximum of total variance at the forward corresponds to negative at-the-forward curvature. Two routes to such shapes are available. A mixture model, or the parameter-randomization framework of Zaugg, Perotti, and Grzelak \cite{ZauggPerottiGrzelak2026}, reaches them with a density that is non-negative by construction, at the cost of working in price space. The induced implied volatility has no closed form and is recovered there through a Taylor expansion in log-moneyness, whose radius of convergence is not known a priori.

Instead, the route taken here is to widen the parametric family in volatility space, permitting implicit bimodal risk-neutral densities and W-shaped expressions while retaining a closed form valid at every strike, with wing asymptotics fixed by two parameters and analytic derivatives at the cost of arbitrage conditions that must be characterized and enforced rather than inherited. WSVI adds bounded shape terms while retaining eSSVI as its $m=0$ base member. 
We next rigorously define this family.

\medskip
\noindent\textbf{The family and our main results.} This paper defines and studies the
following family, which we call \emph{WSVI}. Fix an expiry $T$ and a
forward $F$, and write $k=\ln(K/F)$ for log-moneyness and
$w(k)=\sigma(k)^2T$ for total implied variance. Let $\sigma_0>0$,
$s\in\R$, $c\ge0$, and $a=(a_1,\dots,a_m)\in\R^m$ be parameters. With the
\emph{scale} $\sca:=\sigma_0\sqrt{T}$ and the \emph{normalized strike}
$z:=k/\sca$, WSVI postulates
\begin{equation}\label{eq:wsvi}
  w(k)\;=\;\sca^2 f(z),
  \qquad
  f(z)\;=\;\frac{P(z)}{2}+\sqrt{\frac{P(z)^2}{4}+\frac{c\,z^2}{2}},
  \qquad
  P(z)\;=\;1+s\,z+\sum_{j=1}^{m}a_j\,\varphi_j\!\left(\frac{z}{\lambda_j}\right),
\end{equation} 
where $\varphi_1,\dots,\varphi_m$ are the smooth, bounded, one-sided basis members of Definition~\ref{def:onesided}, each attached to a scale $\lambda_j>0$, an order $n_j\ge1$, and a side $\varepsilon_j\in\{-1,+1\}$, and each vanishing at $z=0$. The dimensionless function $f$ is the \emph{shape}, normalized by $f(0)=1$; $\sigma_0$ is the level and is exactly the at-the-forward implied volatility. At $m=0$ the shape reduces to the eSSVI slice (Remark~\ref{rem:ssvi}), so WSVI contains that family and extends it.

The main results of this paper establish three exact properties of the family \eqref{eq:wsvi}. First, the leading-order wing behavior of the shape is determined by $(s,c)$ alone, and no choice of amplitudes can alter it (Proposition~\ref{prop:asymconst}). This follows from the boundedness of every $\varphi_j$: the added freedom reshapes the interior without changing the leading-order wings. Second, the family is exactly nested. Setting an amplitude to zero reduces an $m$-amplitude member to the corresponding $(m-1)$-amplitude member, so the amplitude count defines a genuine model ladder (Proposition~\ref{prop:nesting}). Third, the real-valued domain is characterized exactly by $c\ge0$ for every amplitude count (Theorem~\ref{thm:domain}).

We would like to emphasize that the third property above does not cost the family its motivating flexibility, because of a central distinction from eSSVI: $c$ is \emph{not} the at-the-forward curvature of the shape. By Proposition~\ref{prop:at-the-forward}, 
\[ 
    f''(0)=c+\sum_{j:n_j=1}\frac{a_j}{\lambda_j^2}, 
\] 
so $f''(0)<0$ remains attainable with $c>0$ through negative order-one amplitudes (Remark~\ref{rem:cnotcurv}). Corollary~\ref{cor:at-the-forwardbudget} gives the lower bound on $f''(0)$ imposed by the butterfly condition in terms of skew and level. Empirical assessment of how often market smiles require this additional flexibility is outside the scope of this paper. What is established here is that the family can express shapes that the canonical three-parameter slice cannot.

The rest of the paper is organized as follows. Sections~\ref{sec:curve} and~\ref{sec:asym} develop the family, its exact derivatives, its at-the-forward identities, its correspondence with eSSVI, and its domain. Section~\ref{sec:arb} gives the butterfly, spread, and calendar conditions in shape coordinates. The butterfly condition is level-dependent, while Theorem~\ref{thm:levelsens} identifies the level-free ray factor governing sensitivity to the total-variance level. Theorem~\ref{thm:caldecomp} separates calendar variation into level and shape-drift contributions, Theorem~\ref{thm:sharedshape} provides an exact calendar-arbitrage-free construction for expiries sharing one shape, and Theorem~\ref{thm:calbudget} bounds the shape-drift contribution when the shape is not shared.
\newpage

\noindent \textbf{Conventions and state variables.}
We summarize the symbols used throughout this paper below:

{\footnotesize
\renewcommand{\arraystretch}{0.95}
\begin{center}
\begin{tabular}{@{}l@{\hspace{0.6em}}p{5.6cm}@{\hspace{1.4em}}l@{\hspace{0.6em}}p{5.0cm}@{}}
\toprule
Symbol & Meaning & Symbol & Meaning\\
\midrule
$k$ & $\ln(K/F)$, log-moneyness (Def.~\ref{def:logmoneyness}) & $\chi$ & $s^2/4+c/2$, radicand coefficient\\
$w,w',w''$ & total implied variance, $k$-derivatives & $C_{\pm}$ & asymptotic wing slopes \eqref{eq:wingslopes}\\
$\sigma_0$ & at-the-forward volatility, a parameter & $R$ & radicand $A^2+Q$; the domain is $R\ge0$\\
$\sca$ & $\sigma_0\sqrt{T}$, the scale \eqref{eq:scale} & $\delta_c$ & strictly positive floor on $c$ (Cor.~\ref{cor:admissible})\\
$z$ & $k/\sca$, normalized strike \eqref{eq:normstrike} & $g$ & density factor \eqref{eq:gdef}; $g\ge0$ is butterfly\\
$\theta$ & $w(0)=\sca^2$, at-the-forward total variance & $\psi$ & density weight \eqref{eq:psidef}\\
$f,f',f''$ & shape function \eqref{eq:fdef}, $z$-derivatives & $\dens$ & implied density of $k$ \eqref{eq:density}\\
$P,Q,A,B$ & shape constituents, $f=A+B$ & $\Cdf$ & implied survival fn.\ \eqref{eq:cdfdef}; $0\le\Cdf\le1$\\
$s,c$ & dimensionless skew $s=f'(0)$; radicand parameter $c$ & $\Ray$ & ray factor $1-kw'/(2w)$ \eqref{eq:raydef}\\
$a_j,\lambda_j,n_j,\varepsilon_j$ & amplitude, scale, order, side of member $j$ & $v$ & local variance \eqref{eq:localvar}\\
$m,M$ & amplitudes in use; maximum available & $\bpi$ & shape parameter vector $(s,c,a)$\\
$\varphi_n^{\varepsilon}$ & one-sided basis member \eqref{eq:phidef} & $\vartheta,\rho_*,\phi,\Psi$ & eSSVI coordinates, $\Psi=\vartheta\phi$\\
\bottomrule
\end{tabular}
\end{center}
}

The remainder of this section fixes the coordinates in which the family is written. Fix one underlying and one expiry. The following are taken as given. Their construction from listed quotes, including the recovery of $F$ from put--call parity and the removal of the early-exercise premium on American options, is a separate problem, and the family below is agnostic to the method used.

\begin{itemize}\itemsep2pt
  \item $T\in\R_{>0}$, the time to expiry in years, on whichever clock the application adopts.
  \item $F\in\R_{>0}$, the forward of the underlying to $T$.
  \item $K\in\R_{>0}$, a strike.
  \item $r\in\R$, a continuously compounded discount rate to $T$.
\end{itemize}

\begin{definition}[Log-moneyness]\label{def:logmoneyness}
The \emph{log-moneyness} of a strike $K$ relative to the forward $F$ is
\[
    k \;:=\; \ln(K/F)\;\in\;\R .
\]
\end{definition}

Strikes are measured relative to the forward rather than in absolute terms, so that slices of different underlyings and expiries are expressed on a common coordinate. The map $K\mapsto k$ is a bijection $\R_{>0}\to\R$; $k=0$ corresponds to $K=F$, and $k$ is the only strike coordinate used below.

\begin{definition}[Total implied variance]\label{def:totvar}
Let $\sigma(k)>0$ denote the Black implied volatility at log-moneyness $k$. The \emph{total implied variance} is $w(k):=\sigma(k)^2\,T>0$.
\end{definition}

We note that derivatives with respect to $k$ are denoted as $w'=\mathrm{d}w/\mathrm{d}k$ and $w''=\mathrm{d}^2w/\mathrm{d}k^2$. 
Next, total variance rather than volatility is the modeled quantity throughout. The static no-arbitrage conditions of Section~\ref{sec:arb} are expressed in terms of $w$, $w'$, and $w''$, while the calendar condition compares total variance across expiries. Expressing these conditions directly in volatility would introduce additional maturity factors.

\begin{definition}[Black price]\label{def:black} 
    For a call with strike $K$ and total variance $w$, 
    \[ 
        V_{\mathrm{c}}(K)\;=\;e^{-rT}\Big[F\,\Ndist(d_1)-K\,\Ndist(d_2)\Big], \qquad d_{1,2}\;=\;-\frac{k}{\sqrt{w}}\pm\frac{\sqrt{w}}{2}, 
    \] 
    where $\Ndist$ denotes the standard normal distribution function and $\ndens(x):=(2\pi)^{-1/2}e^{-x^2/2}$ its density. Puts follow from put--call parity, $V_{\mathrm{c}}-V_{\mathrm{p}}=e^{-rT}(F-K)$. 
\end{definition}
$d_1$ and $d_2$ depend on $(K,F,w)$ only through $k$ and $w$. The construction is \emph{forward-relative}. Once the relevant forward level $F$ is specified, its method of construction does not otherwise enter the parameterization. When $F$ is a quoted futures price, Definition~\ref{def:black} reduces to the Black (1976) formulation. The cross-expiry theory of Section~\ref{sec:calendar} assumes that all expiries reference the same underlying. 
Two positivity requirements bound the family's scope. The log-moneyness coordinate requires $K>0$ and $F>0$. If the relevant forward can be non-positive, or if the market uses an additive rather than a multiplicative volatility convention, this coordinate is not available and the family does not apply without reformulation.

\section{The parametric family}\label{sec:curve}

This section defines the WSVI parameterization, its basis functions, and the derivatives used throughout the paper. The closed-form derivatives of Section~\ref{sec:exactderiv} are what allow the arbitrage conditions of Section~\ref{sec:arb} to be evaluated exactly rather than by finite differencing, and the at-the-forward identities of Proposition~\ref{prop:at-the-forward} are what give the parameters their interpretation.

\subsection{Level--shape factorization}\label{sec:levelshape}

\begin{definition}[Scale, normalized strike, shape]\label{def:scale}
Let $\sigma_0>0$ be a parameter, the \emph{at-the-forward volatility}.
Define the \emph{scale}
\begin{equation}\label{eq:scale}
  \sca \;:=\; \sigma_0\sqrt{T} \;>\; 0,
\end{equation}
the \emph{normalized strike}
\begin{equation}\label{eq:normstrike}
  z \;:=\; \frac{k}{\sca}\;\in\;\R,
\end{equation}
and postulate that total variance factorizes as
\begin{equation}\label{eq:factorization}
  {\;w(k)\;=\;\sca^2 f(z),\qquad z=k/\sca,\;}
\end{equation}
where $f$ is a dimensionless \emph{shape function} normalized by
\begin{equation}\label{eq:norm}
  f(0)\;=\;1 .
\end{equation}
\end{definition}
Two comments on Definition~\ref{def:scale}. First, the codomain of $f$ is left unspecified: Section~\ref{sec:domain} identifies the parameter region on which the representation is real and strictly positive, and shows that within that region $f:\R\to\R_{>0}$. Second, \eqref{eq:norm} is a normalization rather than an additional constraint, since Proposition~\ref{prop:at-the-forward} shows that the parameterization satisfies it identically. Consequently, $\sigma_0$ alone carries the level.

Substituting \eqref{eq:normstrike} into \eqref{eq:factorization} gives the identity used throughout. Since $z=k/\sca$ and $\mathrm{d}z/\mathrm{d}k=1/\sca$,
\begin{equation}\label{eq:derivid} 
    w \;=\; \sca^2 f(z),\qquad w' \;=\; \sca\, f'(z),\qquad w'' \;=\; f''(z). 
\end{equation} 
Thus, the required $k$-derivatives follow directly from the shape derivatives. The first derivative carries one factor of $\sca$, while the second carries none.
\begin{definition}[At-the-forward total variance]\label{def:theta}
The \emph{at-the-forward total variance} is the value of $w$ at $k=0$,
\[
    \theta \;:=\; w(0)=\sca^2 f(0)=\sca^2=\sigma_0^2 T,
\]
where the last two equalities use the normalization \eqref{eq:norm}.
\end{definition}

Thus $\theta$, $\sca$, and $\sigma_0$ describe the same degree of freedom in different coordinates: total variance, width in $k$-space, and at-the-forward volatility, respectively.

\subsection{The shape function}\label{sec:shapefn} 
\begin{definition}[Shape function]\label{def:shape} Let $s\in\R$ and $c\in\R$ be parameters (Theorem~\ref{thm:domain} restricts $c$ to $c\ge0$; the wider range is kept here so that the restriction is proved rather than assumed), let $m\in\{0,1,\dots,M\}$ be an \emph{amplitude count}, and let $a:=(a_1,\dots,a_m)\in\R^m$ denote the amplitudes. Let $\varphi_1,\dots,\varphi_m:\R\to\R$ be the basis functions of Definition~\ref{def:onesided}, with attached scales $\lambda_j>0$.
Define
\begin{align}
  P(z) &\;:=\; 1 \;+\; s\,z \;+\; \sum_{j=1}^{m} a_j\,\varphi_j\!\left(\frac{z}{\lambda_j}\right),
  \label{eq:Pdef}\\[2pt]
  Q(z) &\;:=\; \tfrac12\,c\,z^2, \label{eq:Qdef}\\[2pt]
  A(z) &\;:=\; \tfrac12 P(z),\qquad B(z)\;:=\;\sqrt{A(z)^2+Q(z)}, \label{eq:ABdef}\\[2pt]
  f(z) &\;:=\; A(z)+B(z)\;=\;\frac{P(z)}{2}+\sqrt{\frac{P(z)^2}{4}+\frac{c\,z^2}{2}}.
  \label{eq:fdef}
\end{align}
The parameter vector is
\begin{equation}\label{eq:pvec}
  p \;:=\; \big(\sigma_0,\,s,\,c,\,a_1,\dots,a_m\big)\in\R_{>0}\times\R^{2+m},
  \qquad \dim p \;=\; 3+m .
\end{equation}
\end{definition}
\begin{remark}[Why a square root rather than a polynomial]\label{rem:root}
The square root in \eqref{eq:fdef} is what makes $f$ asymptotically
\emph{linear} in $|z|$ (Section~\ref{sec:asym}), while remaining smooth and
positive near $z=0$. A polynomial in $z$ of degree $\ge2$ grows too fast to
satisfy Lee's large-strike bound (Proposition~\ref{prop:wingallm}); a
piecewise construction would introduce knots into $w''$ and hence into the
implied density of Section~\ref{sec:butterfly}.
\end{remark}

\subsection{The term basis}\label{sec:basis}

\begin{definition}[Basis primitives]\label{def:primitives}
For $u\in\R$ and an integer \emph{order} $n\ge1$, define
\begin{equation}\label{eq:et}
  e(u)\;:=\;\frac{u^2}{1+u^2},\qquad t(u)\;:=\;\frac{u}{\sqrt{1+u^2}},
\end{equation}
and
\begin{equation}\label{eq:EO}
  E_n(u)\;:=\;e(u)^n=\left(\frac{u^2}{1+u^2}\right)^{\!n},\qquad
  O_n(u)\;:=\;t(u)^{2n+1}=\left(\frac{u}{\sqrt{1+u^2}}\right)^{\!2n+1}.
\end{equation}
\end{definition}
$E_n$ is even and $O_n$ is odd. Their limits are
\begin{equation}\label{eq:EOlim}
  \lim_{|u|\to\infty}E_n(u)=1,\qquad \lim_{u\to\pm\infty}O_n(u)=\pm1 .
\end{equation}

\begin{definition}[One-sided basis members]\label{def:onesided}
For a side $\varepsilon\in\{-1,+1\}$, with $+1$ denoting the \emph{call}
side and $-1$ the \emph{put} side,
\begin{equation}\label{eq:phidef}
    \varphi_n^{\varepsilon}(u)\;:=\;\tfrac12\Big(E_n(u)\;+\;\varepsilon\,O_n(u)\Big).
\end{equation}
Basis member $j$ of a slice is $\varphi_j:=\varphi_{n_j}^{\varepsilon_j}$,
evaluated at $u=z/\lambda_j$ in \eqref{eq:Pdef}, with side
$\varepsilon_j$, order $n_j\ge1$, and scale $\lambda_j>0$.
\end{definition}
Notice that from \eqref{eq:EOlim},
\begin{equation}\label{eq:philim}
    \varphi_n^{+1}(u)\;\longrightarrow\;
    \begin{cases}1,& u\to+\infty\\ 0,& u\to-\infty\end{cases}
    \qquad \text{and} \qquad
    \varphi_n^{-1}(u)\;\longrightarrow\;
    \begin{cases}0,& u\to+\infty\\ 1,& u\to-\infty\end{cases}
\end{equation}
so that each member is a smooth one-sided step, saturating on the wing it names
and vanishing on the other. Moreover, since $e=t^2$, every member has the
closed form
\begin{equation}\label{eq:phibound}
    \varphi_n^{\varepsilon}(u)\;=\;\tfrac12\,t(u)^{2n}\big(1+\varepsilon\,t(u)\big)\;\in\;[0,1),
\end{equation}
bounded uniformly in $u$.

\begin{proposition}[Every member is asymptotically constant]\label{prop:asymconst}
Each $\varphi_n^{\varepsilon}$ tends to a finite constant as
$u\to\pm\infty$. Consequently, $P$ in \eqref{eq:Pdef} is asymptotically
affine in $z$ with a slope exactly $s$, regardless of the amplitudes.
\end{proposition}

\begin{proof}
Immediate from \eqref{eq:philim}: the sum in \eqref{eq:Pdef} converges to
$\sum_{j:\,\varepsilon_j=+1}a_j$ as $z\to+\infty$ and to
$\sum_{j:\,\varepsilon_j=-1}a_j$ as $z\to-\infty$, both finite; $sz$ is the only unbounded term.
\end{proof}

Proposition~\ref{prop:asymconst} is the main design property of the basis. The amplitudes can modify the interior of the smile without changing its asymptotic wing behavior, which remains determined by $(s,c)$.

\begin{proposition}[Effect of the order]\label{prop:order}
Raising $n$ at fixed $\lambda$ translates the transition outward and widens it in absolute terms. Specifically, $E_n$ attains half height at
\begin{equation}\label{eq:halfheight}
  u^2 \;=\; \frac{2^{-1/n}}{1-2^{-1/n}},
\end{equation}
which grows linearly in $n$, so the half-height point grows like $\sqrt{n}$.
\end{proposition}

\begin{proof}
$E_n(u)=e(u)^n=\tfrac12$ gives $e(u)=2^{-1/n}$, and inverting
$e=u^2/(1+u^2)$ yields \eqref{eq:halfheight}. As $n\to\infty$,
$2^{-1/n}=1-(\ln2)/n+O(n^{-2})$, so the right-hand side is $n/\ln2+O(1)$
and $|u|\sim\sqrt{n/\ln2}$.
\end{proof}

\begin{remark}[Role of the order]\label{rem:orderrole}
It is tempting to describe the order as sharpening a transition \emph{in
place}, leaving scale to set location. 
This is, however, false, as Proposition~\ref{prop:order} shows.
The proposition is stated for $E_n$;
the one-sided member $\varphi_n^{\varepsilon}$ of \eqref{eq:phidef} carries
the odd part as well and therefore transitions slightly later. Its quantile positions, computed exactly from \eqref{eq:phidef}, are
\[
\begin{array}{c|ccc|cc}
 n & u_{25\%} & u_{50\%} & u_{75\%} & \text{width } u_{75}-u_{25} & \text{width}/u_{50\%}\\\hline
 1 & 0.685 & 1.151 & 1.959 & 1.274 & 1.107\\
 4 & 1.608 & 2.374 & 3.778 & 2.170 & 0.914\\
 12 & 2.889 & 4.144 & 6.487 & 3.599 & 0.868
\end{array}
\]
with $u_{50\%}/\sqrt{n}$ running $1.151\to1.196$ over that range. The
absolute width \emph{grows}; only the width relative to the member's own
location narrows.
\end{remark}

Order and scale, therefore, do not control independent features.
Both move the transition: rescaling $\lambda$ preserves the shape, whereas increasing $n$ narrows the transition relative to its location.
The index set $\{(\lambda_j,n_j,\varepsilon_j)\}$ is fixed for an underlying and shared across its expiries. Consequently, $a_j$ multiplies the same function on every slice. This common basis is required by the shared-shape construction of Theorem~\ref{thm:sharedshape}. Under $\varepsilon\mapsto-\varepsilon$, for every member, the roles of the two wings exchange, allowing the same family to accommodate underlyings of opposite skew signs.

\subsection{Exact derivatives}\label{sec:exactderiv} 
The shape is differentiable in closed form. This section derives the first and second derivatives needed by the arbitrage conditions in Section~\ref{sec:arb}.

\begin{proposition}[Primitive derivatives]\label{prop:primderiv} With $\Omega:=1+u^2$, 
\begin{equation}\label{eq:ederiv}
  e=\frac{u^2}{\Omega},\qquad
  e'=\frac{2u}{\Omega^2},\qquad
  e''=\frac{2-6u^2}{\Omega^3},
\end{equation}
\begin{equation}\label{eq:tderiv}
  t=\frac{u}{\Omega^{1/2}},\qquad
  t'=\frac{1}{\Omega^{3/2}},\qquad
  t''=\frac{-3u}{\Omega^{5/2}}.
\end{equation}
For $n\ge1$ and $q:=2n+1$,
\begin{align}
  E_n'&=n\,e^{n-1}e',
  &E_n''&=n(n-1)\,e^{n-2}(e')^2+n\,e^{n-1}e'', \label{eq:Ederiv}\\
  O_n'&=q\,t^{q-1}t',
  &O_n''&=q\Big[(q-1)\,t^{q-2}(t')^2+t^{q-1}t''\Big], \label{eq:Oderiv}
\end{align}
and
$\varphi_n^{\varepsilon\,(\ell)}=\tfrac12\big(E_n^{(\ell)}+\varepsilon O_n^{(\ell)}\big)$
for $\ell=0,1,2$.
\end{proposition}

\begin{proof}
Direct differentiation. On the odd side, every exponent is a non-negative integer, since $n\ge1$ gives $q-2=2n-1\ge1$. On the even side, $E_n''$ contains $e^{n-2}$, which at $n=1$ is $e^{-1}$ and is undefined at $u=0$, where $e(0)=0$. The expression is nevertheless correct for $n=1$ because its coefficient $n(n-1)$ vanishes there; read literally, it is a removable $0\cdot\infty$. An implementation must treat $n=1$ as a separate case and evaluate $E_1''=e''$ directly rather than substituting it into the general formula.
\end{proof}

\begin{proposition}[Shape derivatives]\label{prop:shapederiv}
The scaled basis members obey the chain rule
\begin{equation}\label{eq:scaledbasisderiv}
\frac{\mathrm{d}^{\ell}}{\mathrm{d}z^{\ell}}
\varphi_j\!\left(\frac{z}{\lambda_j}\right)
=\lambda_j^{-\ell}\,
\varphi_j^{(\ell)}\!\left(\frac{z}{\lambda_j}\right),
\qquad \ell=0,1,2,
\end{equation}
so the constituents of the shape have derivatives
\begin{equation}\label{eq:Pderiv}
P'=s+\sum_j\frac{a_j}{\lambda_j}\,
\varphi_j'\!\left(\frac{z}{\lambda_j}\right),
\qquad
P''=\sum_j\frac{a_j}{\lambda_j^2}\,
\varphi_j''\!\left(\frac{z}{\lambda_j}\right),
\qquad
Q'=cz,
\qquad
Q''=c.
\end{equation}
Writing $A^{(\ell)}=\tfrac12P^{(\ell)}$ for $\ell=0,1,2$, the shape $f=A+B$ is twice differentiable wherever $B>0$, and its second derivative admits the explicit form
\begin{equation}\label{eq:fppgen}
f''(z)
=\frac{f(z)}{B(z)}\,A''(z)
+\frac{c\big(A(z)-zA'(z)\big)^2}{2B(z)^3}
=\frac{f}{2B}\,P''
+\frac{c\,(P-zP')^2}{8B^3}.
\end{equation}
\end{proposition}

\begin{proof}
Equation \eqref{eq:scaledbasisderiv} is the chain rule for $u=z/\lambda_j$, and \eqref{eq:Pderiv} follows by differentiating \eqref{eq:Pdef} and \eqref{eq:Qdef} term by term.

For $B=\sqrt{A^2+Q}$, differentiate $B^2=A^2+Q$ once to obtain $2BB'=2AA'+Q'$, so
\begin{equation}\label{eq:Bprime}
B'=\frac{A\,A'+\tfrac12 Q'}{B}.
\end{equation}
Differentiating again gives $(B')^2+BB''=(A')^2+AA''+\tfrac12 Q''$, so
\begin{equation}\label{eq:Bpprime}
B''=\frac{(A')^2+A\,A''+\tfrac12 Q''-(B')^2}{B}.
\end{equation}
Since $f=A+B$,
\begin{equation}\label{eq:fABderiv}
f=A+B,\qquad f'=A'+B',\qquad f''=A''+B''.
\end{equation}

To obtain \eqref{eq:fppgen}, substitute $Q=\tfrac12 cz^2$, $Q'=cz$, and $Q''=c$ into \eqref{eq:Bprime}--\eqref{eq:Bpprime} and multiply \eqref{eq:Bpprime} by $B^2$:
\[
B^2B''
=B^2(A')^2+B^2AA''+\tfrac12 cB^2-\big(AA'+\tfrac12 cz\big)^2 .
\]
Using $B^2=A^2+\tfrac12 cz^2$, the terms not containing $A''$ collapse:
\[
B^2(A')^2-A^2(A')^2=\tfrac12 cz^2(A')^2,\qquad
\tfrac12 cB^2-\tfrac14 c^2z^2=\tfrac12 cA^2,
\]
and the cross term is $-czAA'$, so
\[
B^2B''=AA''B^2+\tfrac12 c\big(A^2-2zAA'+z^2(A')^2\big)
=AA''B^2+\tfrac12 c\,(A-zA')^2 .
\]
Hence $B''=AA''/B+c(A-zA')^2/(2B^3)$, and by \eqref{eq:fABderiv}
\[
f''=A''+B''=A''\Big(1+\frac{A}{B}\Big)+\frac{c(A-zA')^2}{2B^3}
=\frac{f}{B}A''+\frac{c(A-zA')^2}{2B^3},
\]
since $1+A/B=(A+B)/B=f/B$. The second form in \eqref{eq:fppgen} is $A=\tfrac12P$.
\end{proof}

Equations \eqref{eq:Bprime}--\eqref{eq:Bpprime} are written in recursive form because $B''$ reuses $B'$, so evaluating the triple $(f,f',f'')$ requires only one square-root evaluation. The recursive form also keeps the evaluated value and its derivatives internally consistent. Substituting a different denominator would produce derivatives that are no longer those of the evaluated $f$.

\begin{proposition}[Values at $z=0$]\label{prop:at-the-forward}
For every admissible parameter vector,
\begin{align}
  f(0)&=1, \label{eq:f0}\\
  f'(0)&=s, \label{eq:fp0}\\
  f''(0)&=c+\sum_{j\,:\,n_j=1}\frac{a_j}{\lambda_j^2}. \label{eq:fpp0}
\end{align}
\end{proposition}

\begin{proof}
From \eqref{eq:et}, $e(0)=t(0)=0$, so $E_n(0)=O_n(0)=0$ and
$\varphi_j(0)=0$ for all $j$; hence $P(0)=1$, $Q(0)=0$,
$A(0)=B(0)=\tfrac12$, and $f(0)=1$. Next, $e'(0)=0$ gives $E_n'(0)=0$, and
$t(0)=0$ with $q-1\ge2$ gives $O_n'(0)=0$, so $P'(0)=s$, $A'(0)=s/2$, and
by \eqref{eq:Bprime} $B'(0)=s/2$, whence $f'(0)=s$. 

For the curvature,
$E_n''(0)$ vanishes for $n\ge2$ and equals $e''(0)=2$ for $n=1$, while
$O_n''(0)=0$ for all $n$; so
$\varphi_n^{\varepsilon\,\prime\prime}(0)=\mathbf{1}[n{=}1]$ independently
of $\varepsilon$, giving $P''(0)=\sum_{j:n_j=1}a_j/\lambda_j^2$ and
$A''(0)=\tfrac12P''(0)$. Substituting into \eqref{eq:Bpprime} at $z=0$,
where $A=B=\tfrac12$, $A'=B'=s/2$, $Q''=c$, the $(s/2)^2$ terms cancel and
$B''(0)=A''(0)+c$; hence $f''(0)=2A''(0)+c=P''(0)+c$, which is
\eqref{eq:fpp0}.
\end{proof}

Proposition~\ref{prop:at-the-forward} gives the parameters a direct interpretation. The parameter $\sigma_0$ is the at-the-forward implied volatility, while $s$ is 
the dimensionless skew with $w'(0)=\sca s$. In contrast, $c$ is \emph{not} the at-the-forward curvature. Equation~\eqref{eq:fpp0} shows that both $c$ and the order-one amplitudes contribute to $f''(0)$, and that only order-one members contribute at the forward.

\begin{proposition}[Exact nesting]\label{prop:nesting}
Setting $a_m=0$ in a member with $m$ amplitudes yields exactly the member with $m-1$ amplitudes. Setting $a_1=\dots=a_m=0$ yields the three-parameter base member
\begin{equation}\label{eq:fbase}
  f_{\mathrm{base}}(z)\;=\;\frac{1+sz}{2}+\sqrt{\frac{(1+sz)^2}{4}+\frac{cz^2}{2}}.
\end{equation}
\end{proposition}

\begin{proof}
Immediate from \eqref{eq:Pdef}: an amplitude enters only through the single
term $a_j\varphi_j$, which vanishes identically when $a_j=0$.
\end{proof}

Since the nesting is exact, the family is a genuine ladder indexed by $m$, and
any selection over $m$ compares nested hypotheses.

\begin{remark}[Correspondence with the eSSVI slice]\label{rem:ssvi}
The base member \eqref{eq:fbase} coincides with the per-slice form of the eSSVI parameterization \cite{HendriksMartini2019}; Klassen \cite{Klassen2016} studies the same three-parameter curve as a per-expiry fit under the name S3. Writing that form with total variance $\vartheta$, correlation $\rho_*\in(-1,1)$, and curvature $\phi>0$, the
correspondence is
\begin{equation}\label{eq:ssvimap}
    \sigma_0=\sqrt{\vartheta/T},\qquad
    s=\rho_*\,\phi\,\sqrt{\vartheta},\qquad
    c=\tfrac12\phi^2\vartheta\,(1-\rho_*^2).
\end{equation}
Thus, the eSSVI slice is recovered exactly at $m=0$, while $m>0$ adds the one-sided basis terms. Section~\ref{sec:basemember} gives the inverse correspondence.
\end{remark}

\section{Asymptotics, wings, and the admissible region}\label{sec:asym}
The bounded basis terms \eqref{eq:phidef} modify the interior of the smile without changing its leading-order wing behavior. We now use this separation to derive the asymptotic wing slopes and then determine the parameter region on which the shape is defined.

\begin{definition}[Radicand coefficient]\label{def:chi} Define \[ \chi:=\frac{s^2}{4}+\frac{c}{2}. \] This quantity is the coefficient of $z^2$ in the leading-order radicand as $|z|\to\infty$. \end{definition}

\begin{proposition}[Linear wings]\label{prop:linwings}
Suppose $\chi>0$. Then, as $z\to\pm\infty$,
\begin{equation}\label{eq:wingslopes}
    f(z)\;=\;C_{\pm}\,|z|+O(1),\qquad
    C_{\pm}\;=\;\sqrt{\chi}\pm\frac{s}{2},
\end{equation}
and in $k$-space
\begin{equation}\label{eq:kwings}
  w(k)\;=\;\sca\,C_{\pm}\,|k|+O(\sca^2),\qquad k\to\pm\infty .
\end{equation}
\end{proposition}

\begin{proof}
By Proposition~\ref{prop:asymconst}, $P(z)=sz+O(1)$, so
$A^2+Q=\chi z^2+O(|z|)$ and, for $\chi>0$, $B=\sqrt{\chi}\,|z|+O(1)$;
hence $f=\tfrac12 sz+\sqrt{\chi}|z|+O(1)$, which is \eqref{eq:wingslopes}
on each side, and \eqref{eq:kwings} follows from $|k|=\sca|z|$.
\end{proof}

\begin{proposition}[Signs of the wing slopes]\label{prop:wingsigns}
Assume $\chi\ge0$. Then $C_{\pm}$ are real-valued with
\begin{equation}\label{eq:wingsumprod}
  C_++C_-\;=\;2\sqrt{\chi}\;\ge\;0,\qquad
  C_+C_-\;=\;\chi-\frac{s^2}{4}\;=\;\frac{c}{2}.
\end{equation}
Consequently, $C_+\ge0$ and $C_-\ge0$ if and only if $c\ge0$; and if
$-\tfrac12 s^2\le c<0$, then exactly one of the two is negative, the one on
the wing opposite to the sign of $s$.
\end{proposition}

\begin{proof}
Two reals with a non-negative sum are both non-negative if and only if their product is
non-negative, which by \eqref{eq:wingsumprod} is $c\ge0$; if $c<0$ the
product is negative, so the two have strictly opposite signs, and the
negative one is $C_-$ when $s>0$ and $C_+$ when $s<0$.
\end{proof}

\begin{remark}[The reality condition]\label{rem:reality}
The reality of the wing slopes is exactly
\begin{equation}\label{eq:realitycond}
  \chi\ge0 \qquad\Longleftrightarrow\qquad c\ge-\frac{s^2}{2},
\end{equation}
which is a strictly weaker condition than the domain \eqref{eq:domain}
below; see Remarks~\ref{rem:basequad} and \ref{rem:counterexample}. Within \eqref{eq:domain}, both slopes are
non-negative by Proposition~\ref{prop:wingsigns}. The linear growth in \eqref{eq:kwings} places $w$ within the $O(|k|)$ asymptotic envelope required by Lee's moment bound \cite{Lee2004}; Proposition~\ref{prop:wingallm} below gives the constant.
\end{remark}

\subsection{The domain}\label{sec:domain}

The shape \eqref{eq:fdef} is real only where its radicand
\begin{equation}\label{eq:radicand}
  R(z)\;:=\;A(z)^2+Q(z)\;=\;\frac{P(z)^2}{4}+\frac{c\,z^2}{2}
\end{equation}
is non-negative. The following settles exactly where that is.

\begin{theorem}[Domain]\label{thm:domain}
Fix an amplitude count $m\ge0$. For any choice of basis members and any amplitudes:
\begin{itemize}\itemsep3pt
    \item[(i)] If $c\ge0$ then $Q\ge0$, hence $R\ge A^2\ge0$ for every
    $z\in\R$; $f$ is real with $f\ge\max(P,0)\ge0$, and $f>0$ everywhere
    when $c>0$.
    \item[(ii)] If $c<0$ then $R$ takes negative values.
\end{itemize}
Consequently $R\ge0$ on all of $\R$ if and only if $c\ge0$, for every
member of the family.
\end{theorem}

\begin{proof}
(i) $Q=\tfrac12 cz^2\ge0$ when $c\ge0$, so $R=A^2+Q\ge A^2\ge0$ and
$B=\sqrt{R}\ge|A|$, giving $f=A+B\ge A+|A|=\max(P,0)\ge0$. If $c>0$ then $R>A^2$ strictly
for $z\neq0$, so $B>|A|$ and $f>0$; at $z=0$, $f=1$ by \eqref{eq:f0}.

(ii) Suppose $c<0$, so $Q(z)=\tfrac12cz^2<0$ for every $z\neq0$. By
\eqref{eq:phibound} the member sum is bounded:
$\big|\sum_j a_j\varphi_j(z/\lambda_j)\big|\le\sum_j|a_j|=:M_a$ for all
$z$. If $s\neq0$, then $|P(z)-1-sz|\le M_a$, so $P(z)\to-\infty$ as
$z\to-\operatorname{sgn}(s)\,\infty$, while $P(0)=1>0$; $P$ is
continuous, so $P(z_0)=0$ for some $z_0\neq0$. There $A(z_0)=0$ and
$R(z_0)=Q(z_0)=\tfrac12 cz_0^2<0$. If $s=0$, then $|P|\le1+M_a$, so $A^2$
is bounded while $Q(z)\to-\infty$ as $|z|\to\infty$; hence $R(z)<0$ for
all sufficiently large $|z|$.
\end{proof}

The necessity argument uses only the boundedness of the basis. The same property that prevents the amplitudes from changing the leading-order wing slopes also makes the domain condition independent of the amplitude count.

\begin{remark}[The base-member quadratic]\label{rem:basequad}
At $m=0$, where $P=1+sz$, the failure is quantitative:
\begin{equation}\label{eq:Rbase}
  R(z)\;=\;\chi z^2+\frac{s}{2}z+\frac14,
  \qquad
  \operatorname{disc}R\;=\;\frac{s^2}{4}-\chi\;=\;-\frac{c}{2},
\end{equation}
so for $c<0$, the radicand is negative exactly between two real roots when $\chi>0$, and for all large $|z|$ when $\chi\le0$. The sign of $\chi$ controls only the leading coefficient of \eqref{eq:Rbase}. It does not control the interior of the radicand, so no lower bound on $\chi$ alone guarantees the domain.
\end{remark}

\begin{remark}[A concrete counterexample]\label{rem:counterexample}
Take $s=-0.55$, $c=-0.05$. Then $\chi=+0.0506>0$, so the wing slopes are real and \eqref{eq:realitycond} holds; yet $\operatorname{disc}R=0.025>0$ and $R<0$ on $z\in(1.154,\,4.278)$, an interval of moderate normalized strike. At the vertex $R=-0.1235$, the shape has no real value.
\end{remark}

\begin{corollary}[Domain and admissible region]\label{cor:admissible}
By Theorem~\ref{thm:domain}, the shape is real-valued exactly on
\begin{equation}\label{eq:domain}
    c\;\ge\;0 ,
\end{equation}
and for $c>0$ it is smooth and strictly positive with $B>0$ on all of $\R$.
\end{corollary}

We call $c>0$ the admissible region. Corollary~\ref{cor:convexity} and everything after it assume it. At $c=0$, the shape is real but $Q\equiv0$, so $f=A+|A|=\max(P,0)$: if $P$ crosses zero, $f$ has a kink there and is zero wherever $P\le0$. In practice, $c\ge\delta_c$ for some $\delta_c>0$ has to be imposed as a bound on the parameter, not a penalty, since by Remark~\ref{rem:basequad}, a bound on $\chi$ alone does not control the interior of the radicand.

\begin{remark}[$c$ is not the at-the-forward curvature]\label{rem:cnotcurv} 
    It is tempting to read \eqref{eq:domain} as forfeiting the reason for using a family of this kind at all. The motivation for admitting more general curvature than a single positive-curvature mode allows is that W-shaped smiles observed around scheduled events can carry a local maximum of total variance at the forward, i.e.\ \emph{negative at-the-forward curvature of the shape}. That is a condition on $f''(0)$.
    By Proposition~\ref{prop:at-the-forward}, 
    \[ 
        f''(0)\;=\;c+\sum_{j\,:\,n_j=1}\frac{a_j}{\lambda_j^2}. 
    \] 
    Thus $c$ and $f''(0)$ are different quantities, and $f''(0)<0$ remains attainable with $c>0$ through negative order-one amplitudes. The restriction $c\ge0$ therefore does not remove the negative at-the-forward curvature that motivates the extension. It removes only parameter choices for which the real-valued shape fails to exist. Allowing $c<0$ is unnecessary for producing negative at-the-forward curvature and instead takes the model outside its real-valued domain. 
\end{remark}

Condition~\eqref{eq:realitycond} is only an asymptotic statement. For $c<0$, it becomes $|s|\ge\sqrt{2|c|}$, which ensures that the linear growth of $P$ dominates the negative quadratic contribution in the wings.

\begin{corollary}[Sources of negative curvature]\label{cor:convexity}
Let $c>0$. Then the second term of \eqref{eq:fppgen} is non-negative, and consequently
    \[   \text{for every }z\in\R:\qquad f''(z)<0 \;\Longrightarrow\; P''(z)<0 . \] 
In particular, at $m=0$ one has $P''\equiv0$ and $P-zP'\equiv1$, so 
\begin{equation}\label{eq:baseconvex}   
    f_{\mathrm{base}}''(z)\;=\;\frac{c}{8\,B(z)^3}\;>\;0   \qquad\text{for all }z\in\R,
\end{equation} 
and the base member is strictly convex on all of $\R$. \end{corollary}  
\begin{proof} 
    By Theorem~\ref{thm:domain}(i), $c>0$ gives $R=A^2+Q>0$ for $z\neq0$ and $R(0)=\tfrac14$; hence, $B>0$ everywhere, so the term $c(P-zP')^2/(8B^3)$ in \eqref{eq:fppgen} is non-negative. Since $f>0$ and $B>0$, the remaining term $fP''/(2B)$ carries the sign of $P''$. At $m=0$, $P=1+sz$ gives $P''\equiv0$ and $P-zP'=1$, leaving \eqref{eq:baseconvex}. 
\end{proof}  
Negative curvature of the shape at any strike, not only at the forward, therefore requires $P''<0$, which by \eqref{eq:Pderiv} can occur only through the amplitudes.

\subsection{Base-member conditions in closed forms}\label{sec:basemember}

At $m=0$, WSVI reduces to the three-parameter member \eqref{eq:fbase}, which is a reparameterization of the per-slice eSSVI form. The correspondence can be inverted explicitly, allowing known eSSVI conditions to be written directly in WSVI coordinates.

\begin{proposition}[Inverse correspondence]\label{prop:invcorr} For the base member $m=0$, assume $c>0$. Then the triplet $(\vartheta,\rho_*,\phi)$ of \eqref{eq:ssvimap} is
recovered from $(\sca,s,c)$ by
\begin{equation}\label{eq:invmap}
  \vartheta=\sca^2,\qquad
  \phi=\frac{2\sqrt{\chi}}{\sca},\qquad
  \rho_*=\frac{s}{2\sqrt{\chi}},\qquad
  \Psi:=\vartheta\phi=2\sca\sqrt{\chi}.
\end{equation}
The hypothesis $c>0$ is necessary and not merely convenient: since
$4\chi=s^2+2c$,
\begin{equation}\label{eq:rhosq}
  \rho_*^2\;=\;\frac{s^2}{s^2+2c},
\end{equation}
so $|\rho_*|<1$ iff $c>0$, with $|\rho_*|=1$ at $c=0$ and
$|\rho_*|>1$ for $c<0$. Assuming only $\chi>0$ is too weak: at
$s=-0.55$, $c=-0.05$ one has $\chi=+0.051>0$ yet $\rho_*^2=1.494$, and
\eqref{eq:invmap} returns $|\rho_*|=1.222$, which lies outside the admissible correlation range.
\end{proposition}

\begin{proof}
From \eqref{eq:ssvimap},
\[
  \chi=\frac{s^2}{4}+\frac{c}{2}
  =\frac{\rho_*^2\phi^2\vartheta}{4}+\frac{\phi^2\vartheta(1-\rho_*^2)}{4}
  =\frac{\phi^2\sca^2}{4},
\]
giving $\phi=2\sqrt{\chi}/\sca$; then
$\rho_*=s/(\phi\sqrt{\vartheta})=s/(2\sqrt{\chi})$ and
$\Psi=2\sca\sqrt{\chi}$.
\end{proof}

\begin{corollary}[The eSSVI correspondence requires $c>0$]\label{cor:cpos}
Since $|\rho_*|<1$, $\phi>0$ and $\vartheta>0$, the correspondence \eqref{eq:ssvimap} always gives $c>0$. Conversely, \eqref{eq:rhosq} shows that the correspondence exists only for $c>0$.
\end{corollary}

The per-slice eSSVI family therefore lies strictly inside the WSVI domain. Moreover, by Corollary~\ref{cor:convexity} the base member is strictly convex on all of $\R$: it cannot dip at any strike, not merely at the forward. That additional expressiveness comes from the WSVI amplitudes rather than from allowing $c<0$.

Two conditions for that slice are standard \cite{GatheralJacquier2014}: an asymptotic necessary condition on the wing slope and a sufficient condition for the absence of butterfly arbitrage. In the form $\Psi(1+|\rho_*|)\le4$ and
$\Psi^2(1+|\rho_*|)\le4\vartheta$ respectively, substituting
\eqref{eq:invmap} gives
\begin{align}
  \text{wing slope:}&\qquad
  \sca\left(\sqrt{\chi}+\frac{|s|}{2}\right)\;\le\;2, \label{eq:wingdiag}\\[2pt]
  \text{butterfly (sufficient):}&\qquad
  \chi+\frac{|s|\sqrt{\chi}}{2}\;\le\;1. \label{eq:butterflydiag}
\end{align}
Equation~\eqref{eq:wingdiag} is the same statement as
$\sca\max(C_+,C_-)\le2$ with $C_{\pm}$ from \eqref{eq:wingslopes}; it is
the wing-slope envelope Remark~\ref{rem:reality} refers to, now with its
constant. Necessary and sufficient no-butterfly-arbitrage conditions for
this three-parameter slice are given by Klassen \cite{Klassen2016}, in
normalized-strike coordinates; the exact butterfly domain of the
five-parameter SVI slice, and a global arbitrage-free parametrization of
eSSVI surfaces, are given by Martini and Mingone
\cite{MartiniMingone2022,Mingone2022}.

\begin{proposition}[The wing condition holds at every amplitude count]\label{prop:wingallm}
Fix any $m\ge0$, any members, and any amplitudes, and let $c>0$. Then
\eqref{eq:wingdiag} is necessary for the absence of butterfly arbitrage, in the form
\[
  \sca\max(C_+,C_-)\;\le\;2 .
\]
\end{proposition}

\begin{proof}
By Proposition~\ref{prop:asymconst}, the wing slopes \eqref{eq:wingslopes} are
determined by $(s,c)$ alone, so \eqref{eq:kwings} holds with the same $C_\pm$
at every amplitude count. Lee's moment formula \cite{Lee2004} bounds both
asymptotic slopes of $w$ in $|k|$ by $2$, giving $\sca C_\pm\le2$. Finally,
$\sca\max(C_+,C_-)\le2$ is \eqref{eq:wingdiag}, since
$\max(C_+,C_-)=\sqrt{\chi}+|s|/2$.
\end{proof}

Equation \eqref{eq:wingdiag} was obtained above through the correspondence
\eqref{eq:invmap}, which exists only at $m=0$; Proposition~\ref{prop:wingallm}
shows that the resulting condition is not confined to the base member. It is the one no-arbitrage condition in this paper that is uniform in $m$, and it is so for the same reason the wings are: the amplitudes are bounded.

\begin{remark}[Scope of the diagnostics]\label{rem:diagscope}
Condition \eqref{eq:butterflydiag}, by contrast, is a property of the
\emph{three-parameter} member alone. For $m>0$ the amplitudes alter $P$,
hence $f''$, hence the density factor, so it is neither necessary nor
sufficient: $\chi$ and $s$ no longer determine the shape.
It is a per-slice diagnostic, not a substitute for the pointwise
conditions of Section~\ref{sec:arb}, which are what any use of the family
must enforce.
\end{remark}

\section{Static arbitrage conditions in shape coordinates}\label{sec:arb}

The arbitrage conditions in this section apply to any positive, twice-differentiable total-variance function $w(k)$. We invoke the WSVI factorization only when translating those conditions into shape coordinates.

\subsection{Butterfly: the density factor}\label{sec:butterfly}

Absence of butterfly arbitrage is the requirement that the implied risk-neutral density be non-negative. The following two definitions split that density into a factor carrying its sign and a factor that is positive wherever $w>0$, so that the condition reduces to a statement about the first alone.

\begin{definition}[Density factor]\label{def:densityfactor}
The \emph{density factor} is the dimensionless combination of $w$, $w'$, and $w''$ given by
\begin{equation}\label{eq:gdef}
    g(k)\;:=\;\left(1-\frac{k\,w'(k)}{2\,w(k)}\right)^{\!2}
    -\frac{w'(k)^2}{4}\left(\frac{1}{w(k)}+\frac14\right)
    +\frac{w''(k)}{2}.
\end{equation}
\end{definition}
\begin{definition}[Density weight]\label{def:densityweight}
The \emph{density weight} is the Black kernel at the strike,
\begin{equation}\label{eq:psidef}
    \psi(k)\; := \;\frac{1}{\sqrt{w(k)}}\exp\!\big(-\tfrac12 d_2(k)^2\big),
    \qquad
    d_2(k) = -\frac{k}{\sqrt{w(k)}}-\frac{\sqrt{w(k)}}{2}.
\end{equation}
\end{definition}
The implied density is the product of the two. The Breeden--Litzenberger density in log-moneyness, its factorization into $g$ and $\psi$, and the non-negativity condition \eqref{eq:butterflycond} are well known:
\begin{proposition}[Risk-neutral density \cite{BreedenLitzenberger1978,GatheralJacquier2014}]\label{prop:density}
    The density of $\ln(K/F)$ implied by the smile is
    \begin{equation}\label{eq:density}
      \dens(k)\;=\;\frac{1}{\sqrt{2\pi}}\,g(k)\,\psi(k).
    \end{equation}
    Absence of butterfly arbitrage on an interval is equivalent to $\dens\ge0$ there and, since $\psi>0$ wherever $w>0$, to
    \begin{equation}\label{eq:butterflycond}
      g(k)\;\ge\;0 .
    \end{equation}
\end{proposition}


In shape coordinates, substituting \eqref{eq:derivid} into \eqref{eq:gdef},
\begin{equation}\label{eq:gshape}
  g\;=\;\left(1-\frac{zf'}{2f}\right)^{\!2}
  -\frac{f'^2}{4}\left(\frac{1}{f}+\frac{\sca^2}{4}\right)+\frac{f''}{2},
\end{equation}
which exhibits a fact used below: $g$ is \emph{level-dependent}, and the dependence is affine. Collecting the one term that carries the level,
\begin{equation}\label{eq:gaffine}
  g\;=\;g_0(z)\;-\;\frac{\theta}{16}\,f'(z)^2,
  \qquad
  g_0(z)\;:=\;\left(1-\frac{zf'}{2f}\right)^{\!2}-\frac{f'^2}{4f}+\frac{f''}{2},
\end{equation}
where $g_0$ depends on the shape alone. Thus $g$ is non-increasing in $\theta$ at fixed shape, with slope $-f'^2/16$: a shape satisfying \eqref{eq:butterflycond} at one level satisfies it at every smaller one, and can fail only as the level rises. In particular, freedom from butterfly arbitrage need only be verified at the largest level a shape is served at; the precise statement follows Theorem~\ref{thm:sharedshape}.

\begin{corollary}[At-the-forward curvature budget]\label{cor:at-the-forwardbudget} At $z=0$, using Proposition~\ref{prop:at-the-forward} in \eqref{eq:gshape}, \begin{equation}\label{eq:g0} g(0)\;=\;1+\frac{f''(0)}{2} -\frac{s^2}{4}\left(1+\frac{\sca^2}{4}\right), \end{equation}
so the butterfly condition at the forward is \begin{equation}\label{eq:g0cond} f''(0)\;\ge\;\frac{s^2}{2}\left(1+\frac{\sca^2}{4}\right)-2 . \end{equation} \end{corollary}

Equation~\eqref{eq:g0cond} gives the exact at-the-forward curvature budget. Negative curvature is compatible with the butterfly condition, but only above a lower bound determined by the skew and the level. Larger skew or larger total variance raises this bound.

Figure~\ref{fig:wshape} gives an explicit member with $f''(0)<0$, a W-shaped smile, and a bimodal implied density. The parameters were chosen by hand, and the plotted range satisfies both the butterfly and spread conditions.

\begin{figure}[ht!]
\centering
\includegraphics[width=\textwidth]{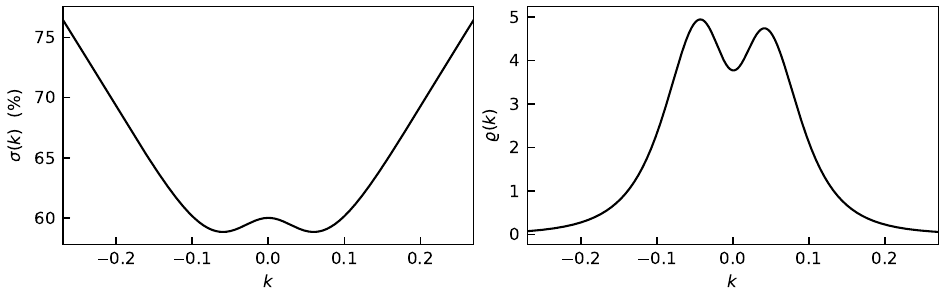}
\caption{A member with negative at-the-forward curvature, strictly inside the domain. Parameters: $\sigma_0=0.60$, $T=5/252$, $s=0$, $c=0.3$, and $m=2$ order-one basis members, one per side ($n_j=1$, $\lambda_j=1$, $\varepsilon_j=\pm1$), each with amplitude $a_j=-0.35$, so that $f''(0)=c+\sum_j a_j/\lambda_j^2=-0.40<0$ while $c>0$, and $g(0)=1+f''(0)/2=0.80$. At these settings the odd parts of the two members cancel and their sum is $-0.35\,e(z)$, so the perturbation is even; the pair is used rather than a single member because the two sides are not tied in general. Left: implied volatility $\sigma(k)$. Right: implied density \eqref{eq:density}, bimodal with modes near $k\approx\pm0.04$. The density factor satisfies $g>0$ on all of $\R$, with $g\to\tfrac14-(\sca C_{\pm})^2/16=0.2499$ in the wings: the member carries no butterfly arbitrage at any strike, and since $\sca C_{\pm}=0.033<2$, Remark~\ref{rem:btfly_impl_sprd_glb} gives the spread conditions at every strike as well. No eSSVI slice can produce this shape as the base member is strictly convex for every admissible parameter choice (Corollary~\ref{cor:convexity}).}
\label{fig:wshape}
\end{figure}

\subsection{Spread: the risk-neutral survival function}\label{sec:spread}
Where the butterfly condition constrains the second derivative of the option-price curve in strike, the spread conditions constrain its first. The corresponding object is the implied survival function.

\begin{definition}[Implied survival function]\label{def:cdf}
    \begin{equation}\label{eq:cdfdef}
        \Cdf(k) \; := \; \Ndist \big(d_2(k)\big) - \ndens \big(d_2(k)\big)\frac{w'(k)}{2\sqrt{w(k)}}.
    \end{equation}
\end{definition}
At fixed volatility, $\Ndist(d_2)$ is the Black probability that the underlying finishes above the strike; with the correction term, $\Cdf$ is the smile-implied survival function: the implied distribution function of $k$ is $1-\Cdf$, and $\Cdf'=-\dens$ by \eqref{eq:density}. The first term of \eqref{eq:cdfdef} is the Black survival function at fixed volatility. The second term accounts for the strike dependence of implied volatility. If the correction is sufficiently large, the implied survival function can leave $[0,1]$.
\begin{proposition}[Spread condition \cite{GatheralJacquier2014}]\label{prop:spread}
    With $V_{\mathrm{c}}$ as in Definition~\ref{def:black} and $k=\ln(K/F)$,
    $\partial_K V_{\mathrm{c}}=-e^{-rT}\Cdf(k)$ and
    $\partial_K V_{\mathrm{p}}=e^{-rT}\big(1-\Cdf(k)\big)$.
    Absence of call- and put-spread arbitrage on an interval of $k$ is therefore equivalent to
    \begin{equation}\label{eq:spreadcond}   
        0\;\le\;\Cdf(k)\;\le\;1 .
    \end{equation}
    Violation below zero is call-spread arbitrage; violation above one is put-spread arbitrage.
\end{proposition}

\begin{remark}[Butterfly implies spread only globally]\label{rem:btfly_impl_sprd_glb}
Since $\Cdf'=-\dens$, the condition $g\ge0$ is exactly the statement that $\Cdf$ is non-increasing. Under the factorization \eqref{eq:factorization} with $c>0$ and $\sca\max(C_+,C_-)<2$, the wings \eqref{eq:kwings} give $\Cdf(-\infty)=1$ and $\Cdf(+\infty)=0$ (at $\sca C_-=2$ exactly, $d_2$ stays bounded as $k\to-\infty$ and the first limit fails), so a non-increasing $\Cdf$ satisfies \eqref{eq:spreadcond} everywhere: $g\ge0$ on all of $\R$ implies the spread conditions. On a bounded interval, the two limits are unavailable. Monotonicity of $\Cdf$ there constrains its variation but not its level, so $g\ge0$ on the interval does not imply \eqref{eq:spreadcond} there. Checking one condition on a finite interval does not remove the need to check the other.
\end{remark}

\subsection{The ray condition and level sensitivity}\label{sec:ray}

To study how a fixed shape responds to changes in its level, define the following factor. It is named for the rays $\theta\mapsto(\theta,z\sqrt{\theta}\,)$ of constant normalized strike, along which the shape is held fixed and only the level varies. 

\begin{definition}[Ray factor]\label{def:rayfactor}
\begin{equation}\label{eq:raydef}
  \Ray(k)\;:=\;1-\frac{k\,w'(k)}{2\,w(k)}.
\end{equation}
\end{definition}
\begin{proposition}[Level-freeness]\label{prop:levelfree}
Under the factorization \eqref{eq:factorization},
\begin{equation}\label{eq:rayshape}
  \Ray\;=\;1-\frac{z\,f'(z)}{2\,f(z)},
\end{equation}
a function of $z$ alone; it depends on the shape and not on the level.
\end{proposition}

\begin{proof}
By \eqref{eq:derivid}, $kw'=\sca^2 zf'$, and $2w=2\sca^2 f$.
\end{proof}

The ray factor also admits a closed form in the shape constituents $P$ and $B$ of \eqref{eq:fdef}; no derivative of $f$ appears in it.

\begin{proposition}[The ray factor in shape constituents]\label{prop:rayP}
Let $c>0$. Then
\begin{equation}\label{eq:rayP}
  \Ray\;=\;\frac12+\frac{P(z)-zP'(z)}{4B(z)} .
\end{equation}
In particular $\Ray>\tfrac12$ if and only if $P>zP'$, and $\Ray\ge0$ if and only if $P-zP'\ge-2B$.
\end{proposition}

\begin{proof}
Since $Q=\tfrac12cz^2$ is homogeneous of degree two, $zQ'=2Q$. Using this together with $B^2=A^2+Q$ and \eqref{eq:Bprime},
\[
  f-zf'\;=\;A+B-z\big(A'+B'\big)
  \;=\;\frac{AB+A^2+Q-zA'B-zAA'-\tfrac12 zQ'}{B}
  \;=\;\frac{(A+B)(A-zA')}{B}\;=\;\frac{f\,(A-zA')}{B}.
\]
Dividing by $2f$ and using $\Ray=1-zf'/(2f)=\tfrac12+(f-zf')/(2f)$ with $A=\tfrac12P$ gives \eqref{eq:rayP}.
\end{proof}

\begin{corollary}[The base member satisfies the ray condition]\label{cor:raybase}
At $m=0$, one has $P-zP'\equiv1$, so
\begin{equation}\label{eq:raybase}
  \Ray_{\mathrm{base}}(z)\;=\;\frac12+\frac{1}{4B(z)}\;>\;\frac12
  \qquad\text{for all }z\in\R .
\end{equation}
\end{corollary}

Hypothesis (i) of Theorem~\ref{thm:sharedshape} therefore holds automatically at $m=0$. A negative ray factor, like negative curvature (Corollary~\ref{cor:convexity}), can occur only through the amplitudes.

\begin{theorem}[Level sensitivity of total variance]\label{thm:levelsens}
Regard $w$ as a function of the level $\theta$ and $k$ through
$w(\theta,k)=\theta f\big(k/\sqrt{\theta}\big)$. Then at fixed $k$,
\begin{equation}\label{eq:levelsens}
  {\;\frac{\partial w}{\partial\theta}\;=\;f(z)-\frac{z}{2}f'(z)
  \;=\;\frac{2w-k\,w'}{2\theta}\;=\;f(z)\,\Ray(k).\;}
\end{equation}
In particular, since $f>0$ wherever $w>0$,
\begin{equation}\label{eq:levelsenssign}
  \frac{\partial w}{\partial\theta}\;\ge\;0
  \qquad\Longleftrightarrow\qquad \Ray\;\ge\;0 .
\end{equation}
\end{theorem}

\begin{proof}
Write $z=k\theta^{-1/2}$, so
$\partial z/\partial\theta=-z/(2\theta)$. Then
$\partial w/\partial\theta=f(z)+\theta f'(z)\big(-z/(2\theta)\big)
=f-\tfrac{z}{2}f'$. The second form follows from \eqref{eq:derivid} and
$k=z\sca$; the third is $f\big(1-\tfrac{zf'}{2f}\big)=f\Ray$ by
Proposition~\ref{prop:levelfree}.
\end{proof}

\begin{remark}[Relation to the density factor]\label{rem:raydensity}
Comparing \eqref{eq:gdef} and \eqref{eq:raydef}, the first bracket of $g$
is exactly $\Ray$:
\[
  g\;=\;\Ray^2-\frac{w'^2}{4}\left(\frac1w+\frac14\right)+\frac{w''}{2}.
\]
Because $g$ carries $\Ray$ \emph{squared}, the condition $g\ge0$ contains no information about the sign of $\Ray$. The ray condition is therefore genuinely independent of the butterfly condition. By \eqref{eq:rayP}, no derivative of $f$ is required to evaluate the ray condition.
\end{remark}

Theorem~\ref{thm:levelsens} depends on the normalization $z=k/\sca$ of \eqref{eq:normstrike}. Under it, $\theta$ enters $w$ only through the prefactor and the argument of $f$, so $\Ray$ carries no dependence on the level.

\subsection{Calendar}\label{sec:calendar}

\begin{definition}[Calendar condition]\label{def:calendarcond}
Let a family of slices be indexed by $T$, each with its own parameter vector. The absence of calendar arbitrage on an interval of $k$ requires
\begin{equation}\label{eq:calendarcond}
  \frac{\partial w(T,k)}{\partial T}\;\ge\;0
\end{equation}
at fixed $k$ when $w$ is differentiable in $T$; equivalently, and without that assumption, $w(T_1,k)\le w(T_2,k)$ for $T_1<T_2$.
\end{definition}
This calendar condition is stated in forward log-moneyness coordinates for expiries of a common underlying.

\begin{theorem}[Calendar decomposition] \label{thm:caldecomp} 
    Let the shape parameters be collected as $\bpi:=(s,c,a)\in\R^{2+m}$, and let $\theta(T)$ and $\bpi(T)$ be differentiable. Then, for a fixed $k$,
    \begin{equation}\label{eq:caldecomp} 
        {\;\frac{\partial w}{\partial T}\;=\; \underbrace{f(z)\,\Ray(k)}_{\partial w/\partial\theta}\;\theta'(T) \;+\;\theta\,\frac{\partial f}{\partial\bpi}\cdot\bpi'(T).\;} 
    \end{equation} 
\end{theorem} 
\begin{proof} 
    $w=\theta f(k/\sqrt{\theta}\,;\bpi)$ depends on $T$ through $\theta$ and $\bpi$; apply the chain rule and Theorem~\ref{thm:levelsens} to the first term. 
\end{proof}

Equation~\eqref{eq:caldecomp} separates calendar variation into a level contribution and a shape-drift contribution. Calendar monotonicity can fail if the level decreases, if the ray factor is negative, or if shape drift overwhelms a non-negative level term. The first two scenarios are pointwise properties that may be imposed exactly, while the third is not.

\begin{theorem}[Shared shape]\label{thm:sharedshape} 
Let a set of expiries $T_1<\dots<T_N$ of one underlying carry a common shape parameter vector $\bpi$ with $c>0$, so that $\bpi'\equiv0$, with levels $\theta_1,\dots,\theta_N$, and let $I\subseteq\R$ be an interval of log-moneyness. Suppose 
\begin{itemize}\itemsep1pt 
    \item[(i)] $\Ray_f(z):=1-\dfrac{zf'(z)}{2f(z)}\ge0$ for every \[ z=\frac{k}{\sqrt{\theta}}, \qquad k\in I,\quad \theta\in[\theta_1,\theta_N], \] and 
    \item[(ii)] $\theta_1\le\theta_2\le\dots\le\theta_N$. 
\end{itemize} 
Then $w(T_i,k)\le w(T_{i+1},k)$ for all $k\in I$ and all $i$. In particular, the slices are calendar-arbitrage-free on $I$. \end{theorem}
\begin{proof} 
    Fix $k\in I$ and index $i$. With $\bpi$ common, \[ w(\theta,k)=\theta f\!\left(\frac{k}{\sqrt{\theta}}\right) \] is a single function of $\theta$. For every $\theta\in[\theta_i,\theta_{i+1}]$, hypothesis (i) gives $\Ray_f(k/\sqrt{\theta})\ge0$. By Theorem~\ref{thm:levelsens}, \[ \frac{\partial w}{\partial\theta}=f\,\Ray_f\ge0, \] since $f>0$ on the admissible region. Therefore \[ w(T_{i+1},k)-w(T_i,k)\;=\;\int_{\theta_i}^{\theta_{i+1}}\frac{\partial w}{\partial\theta}\,\mathrm{d}\theta\;\ge\;0, \] with hypothesis (ii) determining the orientation of the interval. 
\end{proof}
The hypothesis $\bpi'=0$ is a real restriction. A shape common in $z$ places its features at $k=\pm\lambda_j\sca$, so they move with $\sqrt{\theta}$ across expiries; a feature produced by a scheduled event instead sits at a log-moneyness set by the expected move, on both expiries straddling the announcement, and does not scale that way. Theorem~\ref{thm:sharedshape} therefore applies to constructions that impose a common shape. Where the shape varies with maturity, the second term of \eqref{eq:caldecomp} remains, and Theorem~\ref{thm:calbudget} below bounds it.

The guarantee is also one-sided: calendar-freeness is preserved exactly, while the butterfly condition remains level-dependent. By \eqref{eq:gaffine}, one check suffices. If $g\ge0$ at the largest level $\theta_N$ for every $z$ appearing in hypothesis (i), then $g\ge0$ there at every smaller level, and the slices carry no butterfly arbitrage on $I$.

When the shape itself drifts, the remaining term of \eqref{eq:caldecomp} is $\theta\,\partial f/\partial\bpi\cdot\bpi'$. The following theorem bounds it through elementary estimates on the parameter gradient of the shape. The resulting condition is a budget, analogous to the curvature budget of Corollary~\ref{cor:at-the-forwardbudget}: the level must grow fast enough to absorb the worst-case drift of the shape.

\begin{theorem}[Calendar budget for a drifting shape]\label{thm:calbudget}
Let $\theta$ and $\bpi=(s,c,a)$ be continuously differentiable on $[T_1,T_2]$, with $\theta>0$ and $c>0$, and let $I\subseteq\R$ be an interval of log-moneyness. Suppose that for every $T\in[T_1,T_2]$ and every $k\in I$, writing $z=k/\sqrt{\theta(T)}$,
\begin{equation}\label{eq:calbudget}
  \theta'(T)\,f(z)\,\Ray(z)
  \;\ge\;
  \theta(T)\left[\,|z|\,|s'(T)|
  \;+\;\sum_{j=1}^{m}\left|\varphi_j\!\left(\frac{z}{\lambda_j}\right)\right|\,|a_j'(T)|
  \;+\;\frac{|z|\,|c'(T)|}{2\sqrt{2\,c(T)}}\,\right].
\end{equation}
Then $\partial w/\partial T\ge0$ for every $k\in I$, so the slices carry no calendar arbitrage on $I$.
\end{theorem}

\begin{proof}
The parameters $\bpi$ enter \eqref{eq:fdef} through $P$ and $Q$ only, and
\[
  \frac{\partial f}{\partial P}=\frac12+\frac{A}{2B}=\frac{f}{2B},
  \qquad
  \frac{\partial f}{\partial Q}=\frac{1}{2B},
\]
while $\partial P/\partial s=z$, $\partial P/\partial a_j=\varphi_j(z/\lambda_j)$ and $\partial Q/\partial c=\tfrac12 z^2$ by \eqref{eq:Pdef}--\eqref{eq:Qdef}. Hence
\[
  \frac{\partial f}{\partial\bpi}\cdot\bpi'
  \;=\;\frac{f}{2B}\left[z\,s'+\sum_{j}\varphi_j\!\left(\frac{z}{\lambda_j}\right)a_j'\right]
  \;+\;\frac{z^2c'}{4B}.
\]
By Theorem~\ref{thm:domain}(i), $c>0$ gives $B\ge|A|$, so $f/(2B)=\tfrac12+A/(2B)\in[0,1]$; and $B\ge\sqrt{c/2}\,|z|$ gives $z^2/(4B)\le|z|/(2\sqrt{2c})$. The modulus of the display is therefore, by the triangle inequality, at most the bracket in \eqref{eq:calbudget}, and substituting into \eqref{eq:caldecomp} gives $\partial w/\partial T\ge0$.
\end{proof}

Condition \eqref{eq:calbudget} is sufficient, not necessary. The left side is the calendar contribution of the level term; the bracket bounds that of the drifting shape. Where $\theta'>0$, the right-hand side is non-negative and $f>0$, so the condition forces $\Ray\ge0$, and hypothesis (i) of Theorem~\ref{thm:sharedshape} need not be assumed separately. Setting $\bpi'\equiv0$ empties the bracket and returns that theorem along a differentiable level path.

The two theorems are otherwise not comparable. Theorem~\ref{thm:sharedshape} compares a finite set of expiries directly and needs no differentiability in $T$; Theorem~\ref{thm:calbudget} reads $\theta'$ and $\bpi'$, and so applies to an interpolated surface between quoted expiries.

Finally, the corresponding \emph{local variance} is 
\begin{equation}\label{eq:localvar} 
    v(T,k)\;:=\;\frac{\partial w/\partial T}{g(k)}. 
\end{equation} 
This expression is admissible when $\partial_Tw\ge0$ and $g>0$. If $g=0$, it is singular; if $g<0$, the implied density is negative, and the resulting quantity no longer defines an admissible local variance.

\section{Conclusion}\label{sec:conclusion}

WSVI extends the per-slice eSSVI form by a basis of bounded one-sided terms, at the cost of one parameter per term, with both the domain and the wing asymptotics unchanged by the extension. In return, the family can express shapes that the three-parameter slice cannot, strictly within the domain. The wing-slope condition \eqref{eq:wingdiag} survives the extension intact (Proposition~\ref{prop:wingallm}), for the same reason the wings themselves do.

The price of remaining in volatility space is that the absence of butterfly arbitrage is a condition to be imposed rather than a property inherited from the construction. Two questions are left open. The first is an exact parameter-domain butterfly characterization for $m>0$, analogous to those known for the base member. Corollary~\ref{cor:at-the-forwardbudget} gives only the at-the-forward slice of such a characterization. The second is empirical: whether market smiles require additional flexibility and how often. The calibration of the family is treated in a separate forthcoming manuscript by the authors.

\section*{Acknowledgments}

Charles Clevenger and Xiang Wan were partially supported by the National Science Foundation under Grant DMS-2418701. 
Charles Clevenger was also partially supported by the Loyola University Chicago Mulcahy Fellowship.

\bibliographystyle{plain}
\bibliography{wsvi}

\end{document}